\documentclass[a4paper]{article}
\usepackage{fullpage,url,amsmath,amsfonts,amssymb, tedmath}%authblk}
\usepackage{enumitem}

\newcommand{\cc}{\mathbb{C}}

\newcommand{\all}{\beta}
\title{Solving underdetermined systems with error-correcting
  codes }
\author{Ted
 Hurley\footnote{National Universiy of Ireland Galway, email:
 Ted.Hurley@NuiGalway.ie }}
\date{} 
\begin{document}
\maketitle
\begin{abstract}
In an underdetermined system of equations $Ax=y$, where $A$ is an $m\times n$
matrix, only $u$ of the entries of $y$ with $u < m$ are known. Thus
 $E_jw$, called `measurements',  are known for certain $j\in J \subset
 \{0,1,\ldots,m-1\}$ where 
 $\{E_i, i=0,1,\ldots, m-1\}$ are  the rows of $A$ and $|J|=u$. 
It is required, if possible, to solve the system uniquely when  
$x$ has at most $t$ non-zero entries with $u\geq 2t$.  

Here such systems are considered from an error-correcting 
coding point of view. The unknown  $x$ can be shown to be the error
 vector of a code subject to certain conditions on the rows of the
 matrix $A$. This reduces the problem to finding a suitable decoding
 algorithm which then finds $x$.  

Decoding workable algorithms  are shown to exist, from which  the
 unknown $x$ may be determined, in cases where the known
 $2t$ values are 
 evenly spaced (that is, when the elements of $J$ are in arithmetic
 progression) for classes of matrices satisfying certain 
 row properties. These cases include Fourier $n\times n $ matrices
 where the 
 arithmetic difference $k$ satisfies $\gcd(n,k)=1$, and classes
 of Vandermonde matrices $V(x_1,x_2,\ldots,x_n)$ (with $x_i\neq 0$)  
with arithmetic difference $k$ 
where the ratios $x_i/x_j$ for $i\neq j$ are not $k^{th}$ roots of
 unity.    
The decoding algorithm has complexity $O(nt)$ and in some cases,
including the Fourier matrix cases,  the 
 complexity is  $O(t^2)$. % and $O(n\log n)$. %%  and deit is
 %%  shown that  

Matrices which have the property that the determinant of any square
 submatrix is non-zero are of particular interest. %% ; these include
 %% $n\times n$ Fourier matrices where $n$ is prime, Vandermonde real
 %% matrices with positive entries and Cauchy matrices.
 Randomly choosing rows of such matrices can then give $t$
 error-correcting 
 pairs to generate  a `measuring' code $C^\perp=\{E_j | j\in J\}$ %are derived 
% to generate a `measuring'
%  system and  
with  a decoding algorithm which finds $x$. % and hence find the unknown $x$ %% . leads
%  %% to further choices for 'sampling' rows which may be decoded to give %value of
%  %% $x$
%  where it is known that $x$ has at most $t$ non-zero entries.

This has 
applications to signal processing and compressed sensing. 
 \end{abstract}
\section{Introduction}
Underdetermined systems $Aw=y$ are considered where $A$ is an 
$m\ti n$ matrix, $w$ an $n\ti 1$ unknown vector and 
$u$ entries of $y$ are known with $u<m$. It is given that %It must be assumed 
$w$ has at most $t$ non-zero entries and that    $u\geq 2t$.
Thus the vector $w=(\al_1,\al_2, \ldots,
\al_{n})\T $ is  known to have at most $t$ non-zero entries but the positions
and the values of these non-zero entries are unknown. 
 %% initially and certain values of $Aw$ are given or known.

Let the rows of $A$ be denoted by 
$\{E_0,E_1, \ldots, E_{n-1}\}$. Hence  
 $E_jw$ are taken or known for $j \in J= \{j_1, j_2, \ldots, j_u\}$   
where $u\geq 2t$, and the problem is to determine $w$, if possible,
from the `measurements' $\{E_{j}w, j\in J\}$. These measurements are
sometimes referred to as  `samples of $w$'. 

%% The $\{E_jw\, | \, j\in J\}$ will sometimes be referred to as
%% `the measurements'.  % for which there may be  a choice. %  for the rows and
% the matrix.  
% $E_jw$ are taken or known for $j \in \{j_1, j_2, \ldots, j_u$   
% where $u\geq 2t$ at least. 

This has   applications to signal processing and  
compressed sensing for which there is a huge and extensive
literature. A  signal may be measured or
sampled by rows of a matrix. 
The work by 
Cand\`es, Romberg and Tao, \cite{compress}, is a
basic reference for recent  treatments of compressed sensing.  
%% Compressed sensing is a signal processing technique for efficiently
%% acquiring and reconstructing a signal, by finding solutions to
%% underdetermined linear systems. The systems must of course be subject 
%%  to certain
%% constraints; in particular the signal must be sparse in the sense
%%  it is required  that  the number of known values must be
%%  greater than or equal to twice the number of non-zero entries. Compressed
%%  sensing has an extensive literature and
%%  many applications. A basic reference given is \cite{tao} by 
%%  Cand\`es, Romberg and Tao. 

 Here a linear algebra approach is taken
 based on error-correcting codes. It is shown that when the $\{E_j\, |
 \, j\in J\}$ generate a code
 $C^\perp$ such that the distance $d(C)$ of the dual code $C$
 satisfies $d(C)\geq 2t+1$, where $t$ is the maximum number of non-zero
 elements of $w$, then $w$ can be obtained by decoding. The problem
 then is to find a suitable decoding algorithm which is efficient and
 stable.    

A general algorithm,  which is a decoding process in disguise,  
for cases where there exist {\em error-correcting pairs} for $C$ (see section
\ref{pairs} below for definition) %% where the rows of $A$ satisfy a
%% multiplicative closure condition (see Section \ref{solve} ) 
%with required distance properties 
is developed   
in Section \ref{algol}.  In certain cases when the
measurements are evenly spaced,  and with additional 
properties on the rows of $A$,  %and with certain distance requirements, 
error-correcting pairs are explicitly 
shown to exist. In these cases an explicit decoding %workable
algorithms are  given in   
Algorithm \ref{algor} and %in general
in 
Algorithm \ref{algor1}. Evenly spaced here means that the elements in $J$ are in
arithmetic sequence. % For some of these 
% cases it is a requirement that the arithmetic
% difference $k$ satisfies $\gcd(k,n)=1$; this  is always the case when
% $n$ is prime or $k=1$.
 %% the algorithm
%% for the Fourier matrix when $k=1$ is given in Section \ref{algol1},
%% Algorithm \ref{algor2},

When the arithmetic difference $k$ (in $J$) satisfies $\gcd(n,k)=1$
the Fourier $n\times n$ matrix  % and  
% when the measurements have arithmetic difference 
 is shown to
 satisfy the conditions and error-correcting pairs are exhibited which
 then solves systems when $A$ is a Fourier matrix. % explicitly. %l case of this general case.
 An algorithm for the Fourier $n\times n$ matrix, with the proviso
 that $\gcd(n,k)=1$, is then  
 given in Section %This is then applied to find an algorithm in Section
 \ref{algol1}, Algorithm 
\ref{algor3}.    
When $k=1$ (that is, when the measurements
are taken consecutively) the algorithm for the Fourier matrix case  
is similar to that obtained in \cite{FRI}. This paper \cite{FRI} also makes the
point  ``to  make the algorithm more robust to noise we have to 
increase the number of available samples .. and apply some
denoising algorithm to the samples''. %%  but this aspect is not dealt with
%% here. 

The Vandermonde matrix $A=V(x_1,x_2,\ldots, x_n)$, ($x_i\neq 0$), 
in which the quotients
 $x_i/x_j$, for $i\neq j$, are not $k^{th}$ roots of unity,
 where $k$ is the arithmetic
difference in $J$, is shown to satisfy the general requirements and 
error-correcting pairs are explicitly given for these cases. % of the   
 %with which to solve. % and satisfies  the general requirements. %  for the
% explicit existence of
% error-correcting pairs .
 An algorithm for finding a solution of $Aw=y$ for  such a Vandermonde 
matrix $A$ with rows $E_i$ where $\{E_jw | j\in J\}$ are known  
with $J$ in arithmetic sequence with difference $k$ such that
$x_i/x_j$, ($i\neq j, 1\leq i,j \leq n$) is not a $k^{th}$ root of unity), 
is given in Section
 \ref{vander}, Algorithm \ref{algor10}.  

The algorithms in general involve a maximum of $O(nt)$  %% $O(n\log n)$
operations but in some cases  a maximum of $O(t^2)$ operations only is
required.  The Fourier $n\times n$ matrix case  
requires $\max(O(t^2), O(n\log n))$ operations and the operations 
are known to be particularly efficient and stable. 

The technique involves considering the problem as a coding/decoding problem
and then to find suitable decoding algorithms. A particularly useful 
decoding algorithm involves finding  {\em error-correcting pairs} for the
code. The method of error-correcting pairs is due jointly to Pelikaan
\cite{pell1} and Duursma \& 
K\"otter \cite{koetter}. 

A technique is derived in Section \ref{random} 
to deal with a type of random
selection of rows of a matrix which has the property that 
% the additional condition that
the determinant of any square submatrix is non-zero. Matrices which satisfy
the condition that the determinant of any square submatrix is non-zero include
the Fourier $n\times n$ matrices where $n$ is prime (Chebotar\"ev's
Theorem), real 
Vandermonde matrices with positive (distinct) 
 entries and Cauchy matrices. 

\section{Coding theory method}  
Consider %hink of 
$w$ as the {\em error vector of a code}. As $w$ has
at most $t$ non-zero entries, a $t$-error correcting
code for which $w$ is the error vector is then required. A method 
which can locate and identify the `errors', which are then 
 the entries of $w$, solves for $w$.

A basic reference for coding theory is \cite{blahut}. 
A code $C$ over a field $F$ is a subset of $F^n$ and all codes
considered are linear. 
An $(n,r)$ code is a code of length $n$ and dimension $r$ and an
$(n,r,d)$ code is a code of length $n$, dimension $r$ and (minimum)
distance $d$. An $(n,r,\geq d)$ code is a code of length $n$, dimension $r$ and 
distance $\geq d$. 
An mds (maximal distance separable) code is
an $(n,r)$ code of 
distance $(n-r+1)$, that is, an mds code is an $(n,r,n-r+1)$ code. 

\subsection{Rows generating codes}\label{rows} 
Let the rows of an $m\times n$ matrix  $A$ be denoted by 
$\{E_0,E_1, \ldots, E_{m-1}\}$ and assume these are linearly
independent. Measurements are taken of $Aw$, that is
certain $E_jw$ are taken or known for $j \in J= \{j_1, j_2, \ldots,
j_u\}\subset \{0,1,\ldots,m-1\}$    
and it is given that $u\geq 2t$ where $t$ is the maximum number of
non-zero entries of $w$. It is clear it may be assumed  without loss of
generality that $m=n$ as just a subset of the rows of $A$ are used. 

Let $C^\perp = \langle E_{j_1}, E_{j_2}, \ldots,
E_{j_u}\rangle$. Think of $C^\perp$ as a  code and its dual
(orthogonal complement) is denoted by $\mathcal{C}$.  
% Denote the matrix of
% $\mathcal{C}$ by $C$.
 Now $C^\perp$ is an $(n,u)$ code and so $\mathcal{C}$ is an $(n,n-u)$
code which  % Write $\hat{C}$ for this check matrix of $\C$, that is, 
% $\hat{C}$ is the matrix with rows consisting of the elements of
% $C^\per$. Now ${C}$, 
has  an $(n-u)\times n$ generator matrix denoted by $C$.  % by
%is denoted by $C$.  %% The rows of this check matrix consists of the
%% elements of 
%% $C^\perp$.
Thus $v\in \mathcal{C}$ if and only if $E_iv=0$ for each $E_i\in C^\perp$ or
equivalently $v\in \mathcal{C}$ if and only if $\hat{C}v = 0_{u\times 1}$ where
$\hat{C}$  is the $u\times n$ 
matrix with rows consisting of the elements  $\{E_{j_1},
E_{j_2}, \ldots, E_{j_u}\}$. 

Then
$C\hat{C}\T= 0_{(n-u)\times u}$, which is  equivalent to $\hat{C}C\T= 
0_{u\times (n- u)}$, is the set-up for the generator matrix/check
matrix of a code and its dual.  % The code with matrix
% $\hat{C}$ is the dual of the code ${C}$ and we will sometimes 
% write $\hat{C}$ for $C^\perp$. %% Now  is equivalent to
% %% $C\T\hat{C}\T=0_{n-u\times u}$$.

{\em If  $\mathcal{C}$ is a $t$-error correcting code then it may be used to
obtain $w$}, provided of course a practical decoding algorithm is
available. 

Now $\mathcal{C}$  is an $(n,n-u)$ code and is $t$-error correcting
if its distance is $\geq 2t+1$. The maximum distance that
$\mathcal{C}$ can attain is $u+1$. 
 For $u=2t$ this requires $\mathcal{C}$ to be an $(n,n-2t, 2t+1)$
 code, that is, it must be 
 an mds code. Now $\mathcal{C}$ is an mds $(n,n-2t,2t+1)$ code if and
 only if  its dual $C^\perp$, with matrix 
 $\hat{C}$, is an (mds) $(n,2t,n-2t+1)$ code. 
%Suppose then $u=2t$. 
The check matrix, $\hat{C}$, of $\mathcal{C}$ is an
$(2t\times n)$ matrix. Thus $\mathcal{C}$ has distance $2t+1$ if and
only if any $2t$ 
columns of $\hat{C}$ are linearly independent -- see for example
Corollary 3.2.3 in \cite{blahut} for details on this. 
Thus for $u=2t$ it is required
 that any $2t$ columns of $\hat{C}$, which is a $2t\times n $ matrix, 
are linearly independent.   

For $u>2t$ it is required that $\mathcal{C}$ be a $(n,u,\geq (2t+1))$
code. Now $\hat{C}$, an $u\times n$ matrix, is the check matrix of
$\mathcal{C}$ and thus it is required that any $2t$ columns of $\hat{C}$
be linearly independent. If any $u$ columns of $\hat{C}$ are linearly
independent then of course any $2t$ columns are linearly independent and
$\mathcal{C}$ is at least $t$-error correcting.

There are a number of cases where it can be assured that $\mathcal{C}$
is $t$-error correcting. 
%% As shown in \cite{hur1} when $A$ is the Fourier $n\times n$ matrix
%% when $n$ is a prime this is always automatic as then  any choice of
%% rows or columns generates an mds code. Thus is this case we can be
%% assured that any $2t$ columns of $\hat{C}$ are linearly independent. 

When $A$ satisfies the property that the
determinant 
of any square submatrix of $A$ is non-zero then any choice
of $r$ rows of $A$ 
gives an mds $(n,r,n-r+1)$ code, \cite{hur1}. 
The Fourier $n\times n$ matrix for a prime $n$ has this property by a
result of Chebotar\"ev \footnote{A proof of
this Chebotar\"ev theorem may be found in \cite{isaacs} and proofs
also appear in the expository paper of P.Stevenhagen and H.W Lenstra
\cite{steven}; paper \cite{simple} contains a relatively short
proof. There are several other proofs in the literature some of which
are referred to in \cite{steven}. Paper
\cite{tao} contains a proof of Chebotar\"ev's theorem and refers to it
as `an uncertainty principle'.}. % for finite cyclic groups'.and 
Thus as shown in \cite{hur1} any code obtained
by taking $(n-2t)$ rows of this Fourier matrix
 gives an $(n,n-2t,2t+1)$ mds code. Hence when $A$ is the Fourier
 $n\times n$ matrix for $n$ a prime any 
such $\hat{C}$ has the required mds property. 
A Vandermonde real matrix with positive entries has this
 property, Corollary \ref{pos} below.  

In general if $V(x_1,x_2,\ldots,
 x_n)$ is a Vandermonde matrix and 
 the $E_{j_k}$ in $C^\perp$ are evenly distributed with arithmetic 
 difference $k$ such that the ratios $x_i/x_j$ for all $i\neq j$ are not 
$k^{th}$ roots of unity then $C^{\perp}$ is an mds codes, see Corollary
 \ref{unity} and Section \ref{vander} below. %$A$
 For a general  $n\times n$ Fourier  matrix it will be shown in Section
\ref{Fourier} that mds codes are obtained when 
the $E_j$ in $C^\perp$ are evenly spaced with arithmetic difference $k$
satisfying $\gcd(n,k)=1$. 

When $A$ is a Cauchy matrix, it also has the property that the
determinant of any  submatrix is non-zero but this case can be highly
unstable and a  decoding method is not easy to obtain. 
\subsection{Unit-derived codes}
Suppose $AB=1$ for $n\times n$ matrices $A,B$. Then as shown in
\cite{hur3} taking any $r$ rows of $A$ gives a generator matrix of an
$(n,r)$ code and the check matrix may be obtained by deleting the
corresponding $r$ columns of $B$. Alternatively any $r$ rows of $A$
gives the check matrix of an $(n,n-r)$ code whose generator matrix is
obtained by deleting the corresponding $r$ columns of $B$.

This is the situation we have for the underdetermined given system $Aw=y$
when $A$ is an $n\times n$ matrix with inverse $B$.

\subsection{Decode to solve}
Suppose now that $\hat{C}$ has the required property that any $2t$
columns are linearly independent. Then $\mathcal{C}$ has distance $\geq 2t+1$
and so the code can correct the `errors'; it can find the
elements of $w$ using the check matrix $\hat{C}$. The problem is to
find a suitable decoding method, that is, a method to locate and
quantify these errors. The method should be of reasonable complexity
and stable for applications. 

We show now that  %% in the case of the Vandermonde $n\times n$ 
%% matrix, which includes the Fourier $n\times n$ 
%% matrix,
when the measurements are
evenly spaced within certain matrices %with certain conditions on the arithmetic difference $k$ 
%satisfying $\gcd(n,k)=1$
an error-correcting (decoding) method exists
which identifies $w$. In general the identification can be done in at
worst 
$O(tn)$ operations but in some cases it may be done in at worst 
$O(t^2)$ operations. %% The
%% complexity is at most $O(t^2)$ and not $O(n^2)$
In practical applications $t$ is often much smaller than $n$.  
%% For example if $n=10^7$ and $t=10^4$ then $O(t^2)$ is less than
%% $O(n\log^2n)$ and if $n=10^6$ and $t=10^4$ then $O(t^2)$ is less than
%% $O(n\log^3n)$. 

\subsection{Error-correcting pairs}\label{pairs}

The method of {\em error-correcting pairs}, when they can be shown to
exist,  may be used to locate and
determine the `errors', and these `errors' then determine the elements of $w$. 
The method of error-locating and error-correcting pairs is
due jointly to Pellikaan \cite{pell1} and to Duursma and K\"otter \cite{koetter}.
The method used here  is based mainly on that of Pellikaan \cite{pell1}. 

Let $F$ be a field and $\C$ a (linear) code over $F$. Write 
$n(\C)$ for the code length of $\C$, its minimum distance is denoted 
by $d( \C )$ and denote its dimension by $k ( \C )$. 

Now $w_i$ denotes the $i^{th}$ component
of $w\in F^n$. 
For any  $w \in F^n$  define the support of $w$
by $\supp( w ) = \{ i | w_i \neq  0 \}$
and the zero set of $w$ by $z ( w ) = \{ i | w_i = 0 \}$ . The weight of $w$ is
the number of non-zero coordinates of $a$ and denote it by $wt ( a )$. The
number of elements of a set $I$ is denoted by $| I |$. Thus $wt(a)=
|\supp(w)|$.

We say that $w$ has $t$
errors supported at $I$ if $w = c + e$ with $c \in \C$ and $I = \supp( e
)$ and $| I | = t = d ( w , \C )$. For $\C$ a linear code, the vector
space of $F$ linear functionals on $\C$ is denoted by $\C^\vee$. 

The bilinear form $< , >$ is
defined by $< a , b > = \sum_i a_i b_i$. For a subset $C$ of $F^n$, 
the dual $C^\perp$ of $C$ in $F^n$ with respect to the bilinear form $< ,
>$ is defined by $C^\perp= \{ x | < x , c > = 0,   \forall c \in C \}$.
% , so in this definition Need not to be linear but C ⊥is always linear. 

The sum of two elements of $F^n$ is defined by adding corresponding 
coordinates.  Of use in these considerations is what is termed 
the {\em star multiplication} $a * b$ of two elements $a, b \in F^n$
defined by
multiplying corresponding coordinates, that is $( a* b )_i = a_i
b_i$. For subsets $A$ and $B$
of $F^n$ denote the set $\{ a * b | a \in A, b \in B \}$ by $A * B$. If
$A$ is generated by $X$ and $B$ is generated by $Y$ then $A*B$ is
generated by $X*Y$. 
\begin{Definition} 
%\subsection{More examples ..} 
%\begin{definition} 
Let $\C $ be a linear code in $F^n$. Define the syndrome map
of the code $\C$ by $s : F^n\rightarrow ( \C^\perp )^\vee  , w \mapsto 
 ( v \mapsto < v , w > )$.

\end{Definition}
%\begin{dref} %\end{definition}
%\end{dref} 
%{Definition 2.1:} 

For a received word
$w \in F^n$ we call $s ( w )$ the syndrome of $w$ with respect to the code $C$.

%Remark 2.2 Note that w is a codeword of C if and only if s ( w ) = 0. If
%w is a word with error e , that is to say w = c + e with c ∈ C , then s
%( w ) = s ( e ).  

\begin{Definition} Let $A, B$ and $C$ be linear codes in $F^n$.
 Define the {\em error locator map} $E_w$ of a received word $w$ with respect to
the code $C$  by $E_w : A \rightarrow B^\vee ,  
a \mapsto ( b \mapsto < w , a*b > )$.
\end{Definition}
Remark: If $A * B \subseteq C^\perp$ and $w$ is a word with error $e$, 
then $E_w = E_e$.

\begin{Definition} % Definition 2.5 
Suppose $I = \{ i_1, i_2, \ldots, i_t \}$ , where $1 \leq
i_1 < . . . < i_t \leq  n$ . Let $A$ be a linear code in $F^n$ . Define 
$A ( I ) = \{ a \in A | a_i = 0,  \forall i \in I \}$.  
\end{Definition}
\begin{Definition}
Define the projection map $\pi_I : F^n \rightarrow F^t$ by $ \pi_I ( w
) =  ( w_{i_1} ,\ldots, w_{i_t})$. Define $A_I = \pi_I(A)$.
Let $e \in F^n$. Denote $\pi_I ( e*A )$ by $eA_I$. 
\end{Definition}
\begin{Definition}
Suppose $I= \{ i_1, i_2, \ldots, i_s\}$. Define the inclusion map
$i_I: F^t\rightarrow F^n$
by mapping the $j^{th}$ component, $w_j$ of $w$ into the
$i_j^{th}$ coordinate for all $j=1,2,\ldots,t$ and zeros everywhere
else.

\noindent Define the restricted syndrome map
$s_I: F^t \rightarrow (C^\perp)^\vee$ 
by $s_I=s*i_I$. 

\end{Definition}
\begin{Definition}\label{error}
Let $A, B$ and $C$ be linear codes in $F^n$. We call $( A, B
)$ a {\em $t$-error correcting pair for $C$} if 
\\ 1) $A * B \subseteq C^\perp $ \\ 2) $k ( A ) > t$ \\  3) $d
( A ) + d ( C ) > n$, \\ 4) $d ( B^\perp) > t$. % then we call $( A, B )$
%a $t$-error correcting pair for $C$ .  

\end{Definition}

\begin{Definition}\label{maps}For an element $w\in F^n$ define $E_w : A
  \rightarrow B^\vee, a\mapsto (b\mapsto <w,a*b>)$. 
\end{Definition}
Now refer to the paper \cite{pell1} and in particular Proposition 2.11
therein. The paper contains the following algorithm, Algorithm
2.3, for locating
and determining the values of errors in the code $C$ when
error-correcting pairs exist for $C$:

\begin{Algorithm}\label{pell} (see \cite{pell1}, Algorithm 2.3:):
\begin{itemize}
\item[1.1] Compute $\ker( E_w )$.  
\item[1.2] If $\ker( E_w ) = 0$, then goto 3.2.  
\item[1.3] If $\ker (E_w ) \neq  0$,
then choose a nonzero element $a \in \ker( E_w )$. 
\item[] Let $J = z ( a )$.
\item[2.1] Compute the space of solutions of $s_J ( x ) = s ( w )$.  
\item[2.2] If $s_J( x ) = s ( w )$ has no or more than one solution then
 goto 3.2. \item[2.3] If
$s_J( x ) = s ( w )$ has the unique solution $x_0$ , then compute $wt ( x_0
)$. 
\item[2.4] If $wt ( x_0 ) > t$ , then goto 3.2.  
\item[3.1] Print: ``The
received word is decoded by''; Print: $w-i_J ( x_0 )$; goto 4.
\item[3.2] Print: ``The received word has more than $t$ errors.''  
\item[4] End.

\end{itemize}

\end{Algorithm}
In our case the actual errors are the values required and so 3.1 is
changed accordingly. Case 3.2 will not arise as by assumption $w$ has
at most $t$ non-zero entries or else it will show up pointing out an
error in this assumption.

\section{Solve the system of equations by decoding}\label{solve}
 Recall the star product $u * v$ of two vectors $u, v \in F^n$. This is
 defined by 
multiplying corresponding coordinates, that is $( u* v )_i = u_iv_i$. 
For subsets $U$ and $V$
of $F^n$ denote the set $\{ u * v | u \in U, v \in V \}$ by $U * V$. %% If
%% $A$ is generated by $X$ and $B$ is generated by $Y$ then $A*B$ is
%% generated by $X*Y$. 

Consider now a matrix $A$ with rows $\{E_0,E_1, \ldots,
E_{n-1}\}$. Assume  the matrix $A$ satisfies conditions 
\ref{eqs1} and \ref{eqs2} as follows:
\begin{enumerate}[label=(\alph*),ref=(\alph*)]
\item $E_i*E_j =  E_{i+j}$ for $i+j\leq (n-1)$. 
% for a scalar $\al_{ij} $ which could depend
% on $i,j$. 
      \label{eqs1}
\item Let  $J\subset \{0,1,\ldots,n-1\}$ be in arithmetic
 sequence with $|J|=r$. Then the code generated by $
\{E_j, j\in J\}$ is an mds $(n,r,n-r+1)$ code.\label{eqs2}  
\end{enumerate}

In the above condition \ref{eqs1} it is required that $i+j\leq (n-1)$. 
where $A$ has rows $\{E_0,E_1,\ldots, E_{n-1}\}$. When for example $A$
is the Fourier $n\times n$ matrix then $E_{i+j}$ is always defined with
$E_{i+j}=E_{i+j \mod n}$.  In other cases also $E_i$ may be defined
for all $i\in \Z$ where $E_i, 0 \leq i \leq (n-1)$ correspond to the
rows of $A$ as for example
when $A$ is a Vandermonde matrix. %%  is sometimes the case that vectors 
%% $\{E_i \, | \, i\in \Z\}$ may %in addition
%% be defined which coincide with the rows of $A$ for $0\leq i\leq n-1$. This is the case when $A$ is a Vandermonde matrix or in
%% the Fourier matrix $E_i$ for $i\in \Z$ is defined
%% y $E_i=E_{(i \mod n)}$.
In such  cases the conditions \ref{eqs1} and 
 \ref{eqs2} may be  replaced as follows. Let $A$ have rows
 $\{E_0,E_1,\ldots, E_{n-1}\}$ such that $E_{i}$ are defined for
 $i\geq 0$ (which coincide with rows of $A$ for $0\leq i \leq n-1$).
Assume $A$ satisfies conditions \ref{eqs3} and \ref{eqs4} as follows. 

\begin{enumerate}[label=(\Alph*),ref=(\Alph*)]
\item $E_i*E_j =  E_{i+j}$. %for $i+j\leq (n-1)$ 
% for a scalar $\al_{ij} $ which could depend
% on $i,j$.
      \label{eqs3}
\item Any  $J\subset \{0,1,\ldots,n-1\}$ in arithmetic
 sequence is such that the code generated by $
\{E_j, j\in J\}$ is an mds $(n,r,n-r+1)$ code where $|J|=r$.\label{eqs4}
\end{enumerate}

Only a  subset of the rows of $A$ are used in the general
theory. We may assume $A$ has first row $E_0$ by the following
consideration. %as follows.  
Suppose $A$ has  rows numbered $\{E_1,E_2,\ldots
 E_{n-1}\}$ satisfying conditions \ref{eqs1} and \ref{eqs2} or conditions
 \ref{eqs3} and \ref{eqs4} with $1\leq i,j$.  Introduce a first row
 $E_0$ into $A$ where $E_0$ is the $1\times n$ vector consisting of all $1's$; this new matrix will  still
be  referred  to as $A$ and satisfies the required conditions with
$0\leq i,j$.  %the applications.   

Assume then from now on in this section that 
the matrix $A$ satisfies conditions \ref{eqs1} and
 \ref{eqs2} or where appropriate conditions \ref{eqs3} and
 \ref{eqs4}.   Matrices which satisfy conditions \ref{eqs1} and $\ref{eqs2}$ or
 conditions \ref{eqs3} and \ref{eqs4} are
given in subsequent sections.

Rows of $A$ are given 
to form $C^\perp=\{E_j\, | \, j\in J\}$, as in
Section \ref{rows}, 
where now  the $E_{j}$ are
evenly distributed, that is,  $C^\perp = \langle E_i, E_{i+j}, E_{i+2j},
 \ldots E_{i+(2t-1)j}\rangle$. (It is assumed that
 $2t< n$ and  that $E_k$ are defined for $0\leq k
 \leq i+(2t-1)j$.) % and that suffices in $E_j$ are taken$\mod n$.
%\ref{eqs}
%% Assume $E_j$ exist for $j \geq n$ (for example $E_j=E_{j\mod
%% n}$) where the $E_i$ satisfy $E_i*E_j= E_{i+j}$.  

% or else 
% \item[2.] . 
% \end{itemize}
%% include Fourier matrices and Vandermonde
%% matrices with some further restrictions. 
%% The restriction  $i+j\leq n$ can be overcome
%% when it is possible to define $E_k$ for $k\geq n$ 
%% but such cases are omitted for the moment. Cases where $A$ has rows
%% $\{E_1,E_2, \ldots, E_{n}\}$ satisfying $E_i*E_j=E_{i+j}$ may
%% similarly be dealt with.  

 With this set-up it is possible to get a
  $t$-error correcting pair, (see 
 definition \ref{error}), for $\mathcal{C}$ the dual  code of 
 $C^{\perp}$. %  and to set up appropriate decoding algorithm which solves
%   for $w$.
  % when $\gcd(n,j)=1$. 
 In these cases the vector $w$ (from $Aw$ where rows $E_jw, j\in J$ are
 known) which has at most $t$ non-zero entries, may be
 obtained by applying the method of Algorithm \ref{pell} above due to Pellikaan
 \cite{pell1} to give an appropriate implementable algorithm in which to
 find $w$. It will be shown that the solution may be obtained in at
 most $O(tn)$
 operations and in some cases in at most $O(t^2)$ operations. 

 Take initially the case  $C^\perp
= \langle E_1,E_2,\ldots, E_{2t}\rangle$, 
that is,  the $E_i$ are consecutive starting at $E_1$; the more
general case will be dealt with similarly.%  but the
% process works in general. % situation. %% 

\begin{theorem}\label{correct2} Let $C^\perp=\langle E_1,E_2, \ldots,
  E_{2t}\rangle$  %% in $F^n=\langle E_0,E_1, 
%% \ldots, E_{n-1}\rangle$
with $2t\leq n$ and $\mathcal{C}$ is the dual of the code generated by
$C^\perp$.  Define $U=\langle E_1,E_2,
\ldots, E_{t+1}\rangle, 
 V= \langle E_0, E_1, \ldots, E_{t-1}\rangle$. Suppose that $C^\perp, U,V$
 generate mds codes. Then $(U,V)$ is a $t$-error
 correcting pair for $C$.
 \end{theorem} 
\begin{proof} Now $E_i*E_j= E_{i+j}$.  %% where suffices may be  taken$\mod
 %% n$
  Then $A*B \subseteq \langle E_1,
E_2, \ldots, E_{2t}\rangle \subseteq C^\perp$.

Note that a code is an mds code if and only if its dual is an mds code.

%Now by Corollary  \ref{coro} the codes $C,A,B$ are mds codes and so 
Now $C$ is an $(n,n-2t, 2t+1)$ code, $U$ is an $(n,t+1,n-t)$
code,   $V$ is an $(n,t,n-t+1)$ code and $V^\perp$ is an $(n,n-t,t+1 )$
code.
Thus $k(U)= t+1 >t, d(U)+d(C) = (n-t)+(2t+1)= n+t+1 >n, d(V^\perp) =
t+1> t$ and {\em so $(U,V)$ is a $t$-error correcting pair for $C$}
 (see Definition \ref{error}). 
\end{proof}

In the general case we have the following. The proof is similar to the
proof of Theorem \ref{correct2} above. Let $E_0$ be the vector with
all $1^s$ as entries. The suffices $lj$ in the following theorems actually mean
$l*j$, the multiplication of $l$ by $j$.
\begin{theorem}\label{correct4} Let 
$C^\perp = \langle E_i, E_{i+j}, E_{i+2j},
 \ldots, E_{i+(2t-1)j}\rangle$.  %% where suffices may be taken $\mod n$ and
 The dual code of $C^\perp$ is $C$.
%% Suppose a  vector $E_{i-j}$ exists with
%% $E_{i-j}*E_j= E_i$.  
Define $U=\langle E_{i}, E_{i+j}, E_{i+j}, \ldots, E_{i+tj}\rangle, V=
 \langle E_0, E_{1j}, E_{2j}, \ldots, E_{(t-1)j}\rangle$.  %% (where the
 %% suffices $lj$ in $B$ actually mean $l*j$, multiplication of $l$ by $j$).
 
Suppose $C^\perp, U,V$ generate mds codes.
 Then $(U,V)$ is a $t$-error correcting pair for $C$.
\end{theorem} 

In a set-up there may be  more than one error-correcting pairs and it
may be useful to consider others. For example we could interchange 
some of the elements of $U,V$. 
\begin{theorem}\label{correct3} Let 
$C^\perp = \langle E_i, E_{i+j}, E_{i+2j},
 \ldots, E_{i+(2t-1)j}\rangle$.  %% where suffices may be taken $\mod n$ and
 The dual code of $C^\perp$ is $C$.
Suppose a  vector $E_{i-j}$ exists with
$E_{i-j}*E_j= E_i$.  
Define $U=\langle E_{i-j}, E_i, E_{i+j}, \ldots, E_{i+(t-1)j}\rangle, V=
 \langle E_{1j}, E_{2j}, \ldots, E_{tj}\rangle$.  % (where the
%  suffices $lj$ in $B$ actually mean $l*j$, multiplication of $l$ by $j$).
 
Suppose $C^\perp, U,V$ generate mds codes.
 Then $(U,V)$ is a $t$-error correcting pair for $C$.
\end{theorem} 

It is thus noted that there may exist a number of different error-correcting
pairs for the same code.

\begin{example} Let $A$ have rows $E_i$ with $E_i*E_j=E_{i+j}$. 
Denote $E_i$ by $i$ and thus $E_i*E_j=E_{i+j}$ translates to $i*j=i+j$. 
Let $C^\perp = \langle{5,7,9,11,13,15\rangle}$ so that $C$ is $3$-error 
correcting (when $C^\perp$ is mds).
The following are $3$-error-correcting pairs.
\begin{itemize} \item $U=\langle 5,7,9,11\rangle, V= \langle
 0,2,4\rangle$. \item $U=\langle 3,5,7,9\rangle, V=\langle 2,4,6\rangle$. \item
$U = \langle 1,3,5,7\rangle,  V=\langle 4,6,8\rangle $.

\item When $-i$ exist (as for the Fourier matrix) it's clear that
      further error-correcting pairs for $C$ can easily be found. 

\end{itemize}   
\end{example}
 \subsection{Interpretation}
%In the general case define $F_k=E_{i+(k-1)j}$ for $k=0,1,\ldots, t$.
Consider $C^\perp, U, V$ as in Theorem \ref{correct4}. 
Now apply Algorithm \ref{pell} (derived from
\cite{pell1}) using the error-correcting pairs found in Theorem
\ref{correct4}. (Other correcting pairs, as shown can exist, 
 may also be used.) We show that the error 
locations may be obtained from the matrix given in the following
Theorem relative to the
bases $\{E_{i},E_{i+j}, \ldots, E_{i+tj}\}$ for $U$ and $\{\om_0,\om_1, \ldots,
\om_{t-1}\}$ for $V^\vee$, where $\om_i: E_{kj}\mapsto \de_{ik}$ for
$i=0,2,\ldots,t-1$.  
Write $F_k=E_{i+(k-1)j}$ for $k=1,2,\ldots, 2t$. Thus $U$ has basis
$\{F_1, F_2, \ldots, F_{t+1}\}$ and $C^\perp$ has basis $\{F_1,F_2, \ldots,
F_{2t}\}$. 
Let $\al_s= <w,F_{s}> = F_sw = E_{i+(s-1)j}w$ for $s=1,\ldots,2t$ and
these are known. 

Recall, definition \ref{maps}, 
 that $E_w:U \rightarrow V^\vee, u \mapsto (v\mapsto <w,u*v>)$.

\begin{theorem}\label{hank}  

$E_w$ has  the following matrix relative to the basis $\{F_1,F_{2},
  \ldots, F_{t+1}\}$ for $U$ and the basis $\{\om_1,\om_2, \ldots, 
\om_t\}$ for $V^\vee$, where $\om_i: E_{kj}\mapsto \de_{ik}$.

$\begin{pmatrix}\al_1 & \al_2 &\ldots & \al_{t+1} \\ \al_2&\al_3 &\ldots &
  \al_{t+2} \\ \al_3& \al_4 & \ldots & \al_{t+3} \\ \vdots & \vdots &
  \vdots &\vdots \\ \al_t & \al_{t+1} & \ldots & \al_{2t} \end{pmatrix}$
\end{theorem}
\begin{proof}%This is shown as follows.

Now $E_w: U \rightarrow V^\vee, u \mapsto (v\mapsto <w,u*v>)$.

Thus $E_w$ works as following on $F_1$:

 $F_1 \mapsto \begin{pmatrix} E_0 \mapsto <w,
F_1*E_0> \\ E_{j} \mapsto <w, F_1*E_{j}> \\ \vdots \\ E_{(t-1)j}\mapsto 
<w,F_1*E_{(t-1)j}>
\end{pmatrix}=\begin{pmatrix} E_0 \mapsto <w,
F_1> \\ E_{j} \mapsto <w, F_2> \\ \vdots \\ E_{(t-1)j}\mapsto <w,F_t>
\end{pmatrix} = \begin{pmatrix} E_0\mapsto \al_1 \\ E_{j}\mapsto \al_2 \\
		 \vdots \\ E_{(t-1)j}\mapsto \al_t\end{pmatrix}$.

Thus $E_w: F_1\mapsto \al_1\om_1+ \al_2\om_2+ \ldots +\al_t\om_t$.
Similarly $E_w: F_i\mapsto \al_{i}\om_1+ \al_{i+1}\om_2 + \ldots
+\al_{i+t-1}\om_t$.

Hence $E_w: (F_1,F_2,\ldots, F_{t+1}) \mapsto (w_1,w_2, \ldots, w_t)\begin{pmatrix} \al_1 &\al_2 & \ldots & \al_{t+1} \\ \al_2 & \al_3 &
  \ldots & \al_{t+2} \\ \vdots & \vdots & \vdots & \vdots \\ \al_t &
  \al_{t+1} & \ldots & \al_{2t}\end{pmatrix}$

Thus the matrix of $E_w$ relative  to bases $\{F_1, F_2, \ldots,
F_{t+1}\}$ for $U$ and
the basis $\{\om_1, \om_2, \ldots, \om_t\}$ for $V^\vee$ is
$\begin{pmatrix} \al_1 &\al_2 & \ldots & \al_{t+1} \\ \al_2 & \al_3 &
  \ldots & \al_{t+2} \\ \vdots & \vdots & \vdots & \vdots \\ \al_t &
  \al_{t+1} & \ldots & \al_{2t}\end{pmatrix}$.

\end{proof}

The matrix in Theorem \ref{hank} is a Hankel matrix and its kernel in general 
can be obtained in at most $O(t^2)$ operations. Just any non-zero
element of the kernel is required.  
%% However it's even more special as
%% pointed out below.  

It is then required to multiply a non-zero  element of
the kernel of the matrix  by
$(F_1,F_2,\ldots, F_{t+1})$ to get an actual kernel element of the mapping
$E_w$.
 
Suppose then
such an element $a\in \ker E_w$ has been found. Let $J=z(a)= \{j | 
a_j=0\}$  which
is the set of locations of the zero coordinates of $a$. It is now
required to compute the space of solutions of $s_J(x)=s(w)$. Suppose 
$J=\{i_1,i_2, \ldots, i_t\}$ and  let $x\in F^n$. %=(x_1, x_2,\ldots, x_{n}) $.
Then $s_J(x) = s*i_J(x)$. Let $i_J(x)=y$ and suppose now $y=i_J(x)$ is $x_1$ in $i_1$ position,
$x_2$ in $i_2$ position and in general $x_k$ in $i_k$ position and
zeros elsewhere.

Now $s: F^n \rightarrow (C^\perp)^\vee$ is $u \mapsto (v\mapsto
<v,u>)$. A basis for $C^\perp$ is $\{F_1,F_2,\ldots, F_{2t}\}$. 

Hence 
$s: w \mapsto \begin{pmatrix} F_1 \mapsto <F_1,w> &= \al_1
	      \\ F_2 \mapsto <F_2,w> &= \al_2
 \\ \vdots &\vdots \\ 
F_{2t} \mapsto <F_{2t},w> &=\al_{2t} \\

\end{pmatrix}$. 

 %=(x_1, x_2,\ldots, x_{n}) $.
% Now $s_J(x) = s*i_J(x)$. Let $i_J(x)=y$ and suppose $y=i_J(x)$ is $x_1$
% in $i_1$ position,
% $x_2$ in $i_2$ position and in general $x_k$ in $i_k$ position and
% zeros elsewhere. 

Since $F_i\in F^n$, let  
 $F_i=(F_{i,1}, F_{i,2}, \ldots, F_{i,n})$ for $i=1,2,\ldots,2t$. 

Now 

$s_J(x)=s*i_J(x)$ and so:

$$s_J: x \mapsto \begin{pmatrix} F_1 \mapsto <F_1,y> &= &x_1F_{1,j_1}+
		  x_2F_{1,j_2}+
 \ldots + x_tF_{1,j_t} \\ F_2 \mapsto <F_2,y> &=
 &x_1F_{2,j_1}+ x_2F_{2,j_2}+ 
 \ldots + x_tF_{2,j_t}\\ \vdots &\vdots &\vdots \\ 
F_{2t} \mapsto <F_{2t},y> &=
 &x_1F_{2t,j_1}+ x_2F_{2t,j_2}+ 
 \ldots + x_tF_{2t,j_t} \\

\end{pmatrix}$$. 

Hence solving $s_J(x) = s(w)$ reduces to solving the following:
\begin{eqnarray}\label{est}
\begin{pmatrix}F_{1,j_1} & F_{1,j_2} & \ldots &F_{1,j_t}\\ F_{2,j_1} &
 F_{2,j_2} & \ldots &F_{2,j_t} \\ \vdots & \vdots & \vdots &\vdots \\ 
F_{2t,j_1} & F_{2t,j_2} & \ldots
&F_{2t,j_t}\end{pmatrix}\begin{pmatrix}x_1 \\ x_2 \\ \vdots \\ x_t
			\end{pmatrix}
= \begin{pmatrix} \al_1 \\ \al_2 \\ \vdots \\ \al_{2t} \end{pmatrix}   
\end{eqnarray}
The value of $w$ is then the solution of these equations with entries in
appropriate places as determined by $J$. The values of $F_{i,k}$ are
known and in some cases have nice forms. The matrix in
(\ref{est}) can be of a special type (for example, submatrix of
Vandermonde and/or
consisting of roots of unity) enabling practical (easier) calculation of a
solution to equations (\ref{est}). 
%% , the matrix in (\ref{est}) is a Vandermonde matrix and
%% when the matrix $A$ is the Fourier matrix, the matrix in (\ref{est}) is a Vandermonde
%% matrix with entries which are roots of unity.  

\section{Algorithms}\label{algol} Now algorithms are given based on the
results of Section \ref{solve} with which to solve the underdetermined
systems in various cases.  
Suppose $y=Ax$ where 
 $A$ is an 
$n\ti n$, $w$ an $n\ti 1$ unknown vector and 
where $u$ entries of $y$ are known. It is given that %It must be assumed 
$w$ has at most $t$ non-zero entries. % and that  $u\geq 2t$. 
Denote the rows of $A$ by $\{E_0,E_1,\ldots, E_{n-1}\}$ and suppose that
$E_i*E_j = E_{i+j}$. % for some scalar $\al$. % so that 

%% with difference $k$ satisfying $\gcd(n,k)=1$. When $n$ is prime (and
%% $k<n$) or when $k=1$ it is always the case that $\gcd(n,k)=1$. 
 
%% In summary then: $w$ is size $n$ vector with at most $t$ non-zero
%% entries.
Measurements $E_jw$ (values of $y$) are taken or known for $j \in M= \{j_1, j_2,
\ldots, j_u\}\subset \{0,1,\ldots,(n-1)\}$  where $u\geq 2t$. Suppose the measurements satisfy the
conditions of Theorem \ref{correct4} with  
$C^\perp = \langle E_i, E_{i+j}, E_{i+2j},
 \ldots, E_{i+(2t-1)j}\rangle$  %% where suffices may be taken $\mod n$
 and
 $C$ is the dual code of the code generated by $C^\perp$.
%% Define $A=\langle E_{i-j}, E_i, E_{i+j}, \ldots, E_{i+(t-1)j}\rangle, B=
%%  \langle E_{j}, E_{2j}, \ldots, E_{tj}\rangle$ (where it is assumed
%%  that $E_{i-j}$ makes sense). Suppose $C^\perp, A,B$ generate mds codes.
%%  Then $(A,B)$ is a $t$-error correcting pair for $C$.  
%% Suppose the measurements $E_1w, E_2w, \ldots, E_{2t}w$
%% are taken.
 We give  an Algorithm to calculate the value of $w$ under these
 conditions  when the
measurements are in an arithmetic progression (evenly distributed) and
subject to conditions of Section \ref{solve}.

\subsection{Case $k=1$}
We first for clarity give the algorithm when $M=\{1,2,\ldots, u\}$ and
$u=2t$. This
is easier to explain and avoids the complicated notation necessary for
the general case given below. 
%% The results and algorithm obtained in
%% this case, where the measurements are taken consecutively, 
%% are similar to those  in \cite{FRI}.  

The set-up then is that $A$ is an $n\times n$ matrix with rows
$\{E_0,E_1,\ldots, E_{n-1}\}$ and that
measurements $E_iw$ are taken for $i=1,2,\ldots, 2t$. It is assumed
that $w$ has at most $t$ non-zero entries. 
Then $w$ is  determined as follows:
Let $\al_i=<w,E_i> = E_iw$ for $i\in J=\{1,2,\ldots,2t\}$. 
% Let $\be_j = <w,E_j> = E_jw$ for $j=\in J$  and define
% $\al_i=\be_{j_i}$. 

\begin{Algorithm}\label{algor}

\quad
\begin{itemize}

\item Find a non-zero  element $x\T$ of the kernel of 
$E=\begin{pmatrix} \al_1 &\al_2 & \ldots & \al_{t+1} \\ \al_2 & \al_3 &
  \ldots & \al_{t+2} \\ \vdots & \vdots & \vdots & \vdots \\ \al_t &
  \al_{t+1} & \ldots & \al_{2t}\end{pmatrix}$.
\item Let $a= (E_1,E_2,\ldots, E_{t+1})x\T$. (Any non-zero multiple of $a$
      will suffice as we are only interested  in the zero entries of
      $a$. Note that $a$ is a $1\times n$ vector.)
\item Let $z(a)= \{j | 
a_j=0\}$  which is the set of locations of the zero coordinates of
$a$. Suppose $z(a)=\{j_1,j_2, \ldots, j_t\}$ and denote this set by $J$.
\item Solve $s_J(x) = s(w)$. This reduces to solving the following.
Here $E_i=(E_{i,1}, E_{i,2}, \ldots, E_{i,n})$.
\begin{eqnarray}\label{est2}
\begin{pmatrix}E_{1,j_1} & E_{1,j_2} & \ldots &E_{1,j_t}\\ E_{2,j_1} &
 E_{2,j_2} & \ldots &E_{2,j_t} \\ \vdots & \vdots & \vdots &\vdots \\ 
E_{2t,j_1} & E_{2t,j_2} & \ldots
&E_{2t,j_t}\end{pmatrix}\begin{pmatrix}x_1 \\ x_2 \\ \vdots \\ x_t
			\end{pmatrix}
= \begin{pmatrix} \al_1 \\ \al_2 \\ \vdots \\ \al_{2t} \end{pmatrix}   
\end{eqnarray}

%% $$B: \begin{pmatrix}\om^{i_1} & \om^{i_2} & \ldots &\om^{i_t}\\ \om^{2i_1} &
%%  \om^{2i_2} & \ldots &\om^{2i_t} \\ \vdots & \vdots & \vdots \\ 
%% \om^{ti_1} & \om^{ti_2} & \ldots
%% &\om^{ti_t}\end{pmatrix}\begin{pmatrix}x_1 \\ x_2 \\ \vdots \\ x_t
%% 			\end{pmatrix}
%% = \begin{pmatrix} \al_1 \\ \al_2 \\ \vdots \\ \al_t \end{pmatrix}$$   

\item 
The value of $w$ is then the solution of these equations with entries in
appropriate places as determined by $J$.
\end{itemize}
 
\end{Algorithm}
The complexity of the operations is discussed in Section \ref{general}.
%% Complexity? Finding the kernel of $E$ is of $O(t^2)$ as what is
%% involved is a $t\ti t+1$ Hankel-type matrix. There exist algorithms of
%% $O(n\log^2n)$ with which to solve $E$. 
%% %% What is required of course is just
%% %% any element of the kernel and then find its zeros.
%% The kernel has dimension $1$ in order to satisfy the given conditions. 

%% Solving $B$ in a
%% straightforward manner is of $O(t^3)$. 
%% However the matrix is of a
%% special type and can be inverted in $O(t^2)$ stable operations. 
%% %type and can be dealt with as in \cite{blahut}. 
%% The matrix is a Vandermonde matrix involving roots of unity. There is a
%% formula for the inverse of such a matrix. This formula does reduce to
%% the Forney formula  \cite{blahut} but this would involve calculating
%% $\Delta(x)$ which is not required as we have the locations from $E$
%% already.  
\subsection{General case}\label{general} 
%\section{Algorithm}
Suppose 
 $A$ is an 
$n\ti n$   %  $w$ an $n\ti 1$ unknown vector and 
% where $u$ entries of $y$ are known. It is given that %It must be assumed 
% $w$ has at most $t$ non-zero entries. % and that  $u\geq 2t$.
 matrix with rows $\{E_0,E_1,\ldots, E_{n-1}\}$ satisfying $E_i*E_j =
 E_{i+j}$.\footnote{More generally it is sometimes enough that $E_i*E_j = \al E_{i+j}$
 for some scalar $\al$ but this is not considered here.} % so that 

%% In summary then: $w$ is size $n$ vector with at most $t$ non-zero
%% entries.
Measurements $E_jw$ are taken or known for $j \in J= \{j_1, j_2,
\ldots, j_u\}$  where $u\geq 2t$. The elements in $J$ are in arithmetic
progression with difference $k$ so that the satisfying $\gcd(n,k)=1$. 
%% Suppose the measurements $E_1w, E_2w, \ldots, E_{2t}w$
%% are taken.
Then $w$ is calculated by the following algorithm.
% We give  an Algorithm to calculate the value of $w$ when the
% measurements are in an arithmetic progression (evenly distributed)
% with difference $k$ satisfying $\gcd(n,k)=1$.

% \subsection{Case $k=1$}
% We first for clarity give the algorithm when $J=\{1,2,\ldots, u\}$. This
% is easier to explain and avoids the complicated notation necessary for
% the general case given below.  
% The set-up then is that $A$ is the Fourier $n\times n$ matrix with rows
% $\{E_0,E_1,\ldots, E_{n-1}\}$ and that
% measurements $E_iw$ are taken for $i=1,2,\ldots, 2t$. It is assumed
% that $w$ has at most $t$ non-zero entries. 
% Then $w$ is  determined as follows:
%Let $\al_i=<w,E_i> = E_iw$ for $i\in J=\{1,2,\ldots,2t\}$. 
Let $\al_k = <w,F_{j_k}> = F_{j_k}w$ for $j_k\in J$.   % and define
% $\al_i=\be_{i}$. 
Define $F_i=E_{j_i}$ for $j_i \in J$ %% where we assume different
%% meanings for $E_k$ and $E_{k_l}$.
and $F_0=E_{j_1-k}$ with indices taken$\mod n$.
Let $F_i=(F_{i,1}, F_{i,2},\ldots, F_{i,n})$.  

\begin{Algorithm}\label{algor1}

\quad
\begin{itemize}

\item Find a non-zero  element $x\T$ of the kernel of 
$E=\begin{pmatrix} \al_1 &\al_2 & \ldots & \al_{t+1} \\ \al_2 & \al_3 &
  \ldots & \al_{t+2} \\ \vdots & \vdots & \vdots & \vdots \\ \al_t &
  \al_{t+1} & \ldots & \al_{2t}\end{pmatrix}$.
\item Let $a= (F_0,F_1,\ldots, F_{t})x\T$. (Any non-zero multiple of $a$
      will suffice as we are only interested  in the zero entries of
      $a$. Note that $a$ is a $1\times n$ vector.)
\item Let $z(a)= \{j | 
a_j=0\}$  which is the set of locations of the zero coordinates of
$a$. Suppose $z(a)=\{j_1,i_2, \ldots, j_t\}$ and denote this set by $J$.
\item Solve $s_J(x) = s(w)$. This reduces to solving the following:
\begin{eqnarray}\label{est1}
\begin{pmatrix}F_{1,j_1} & F_{1,j_2} & \ldots &F_{1,j_t}\\ F_{2,j_1} &
 F_{2,j_2} & \ldots &F_{2,j_t} \\ \vdots & \vdots & \vdots &\vdots \\ 
F_{2t,j_1} & F_{2t,j_2} & \ldots
&F_{2t,j_t}\end{pmatrix}\begin{pmatrix}x_1 \\ x_2 \\ \vdots \\ x_t
			\end{pmatrix}
= \begin{pmatrix} \al_1 \\ \al_2 \\ \vdots \\ \al_{2t} \end{pmatrix}   
\end{eqnarray}

%% \begin{eqnarray}\label{matrix}\begin{pmatrix}\om^{i_1} & \om^{i_2} & \ldots &\om^{i_t}\\ \om^{2i_1} &
%%  \om^{2i_2} & \ldots &\om^{2i_t} \\ \vdots & \vdots & \vdots \\ 
%% \om^{ti_1} & \om^{ti_2} & \ldots
%% &\om^{ti_t}\end{pmatrix}\begin{pmatrix}x_1 \\ x_2 \\ \vdots \\ x_t
%% 			\end{pmatrix}
%% = \begin{pmatrix} \al_1 \\ \al_2 \\ \vdots \\ \al_t
%%   \end{pmatrix}\end{eqnarray}

\item 
The value of $w$ is then the solution of these equations with entries in
appropriate places as determined by $J$.
\end{itemize}
 
\end{Algorithm}

%\section{Determinants} 
\section{Random selection}\label{random} 
This section initiates a method for working with randomly chosen 
 error-correcting pairs. It is independent of subsequent  sections. 

Suppose the  $n\times n$ matrix $A$
in the underdetermined system $Aw=y$ has the property that
the determinant of any square submatrix is non-zero. Then the choice of any
$r$ rows of $A$ yields an mds $(n,r,n-r+1)$  code. Matrices which have
this property are the Fourier $n\times n$ matrices with $n$ a prime
(Chebotar\"ev's theorem), the Vandermonde real matrices with positive
entries and Cauchy matrices.  

%% If the
%% rows in $Aw=y$ are chosen to be evenly distributed then a good
%% error-correcting algorithm exists. 
When considering $Aw=y$, if any $r$ rows of $A$ are chosen
for $C^\perp$  (notation as in
Section \ref{rows}) then an mds code for $\C$ is obtained   
but we haven't an error-correcting pair
to hand as when the rows are evenly distributed. Now approach the
randomness from another point of 
view of {\em  choose the error-correcting pair randomly} and this decides the
rows to be chosen for the measurements (code); 
then the randomly chosen pair is an error-correcting pair for this code. 

This section enables working with rows of matrices which have the
property that the determinant of any square submatrix is non-zero
 as the `samples' for $Aw$. However the systems in
general may require more than $2t$ samples when the $w$ has just $t$
non-zero entries. % at worst. %%  but much less than the
%% Nyquist rate for small $t$.

Consider then the following Proposition of  Duursma and K\"otter 
\cite{koetter}.

(For $U,V\in F^n$ let $U*V$ denote the space generated by $\{u*v |u\in
U, v\in V\}$.)

\begin{proposition}\label{duursma} {\em (See Proposition 1 of \cite{koetter}.)}
Let $U,V$ be mds codes with $k(U)=t+1, k(V)=t$. Any code $C\perp U*V$
 has distance $\geq 2t+1$ and has $t$-error correcting pair $(U,V)$.
\end{proposition}

\subsubsection{Illustrative examples of random selection}
The examples given necessarily have small length so they can be
displayed but in general large length examples are easily obtained.  % but getting large length such is easily attained. 

\begin{itemize}\item {Example 1:} Let $n=19$ and let 
$A$ be the $19\times 19$ Fourier
matrix with rows $\{E_0,E_1,\ldots, E_{18}\}$. As $19$ is prime any choice
of rows of $A$ gives an mds code. We now manufacture a $3$-error 
correcting code ($t=3$) with $3$-error correcting pair.
Then randomly choose $4$ and $3$ rows of $A$. Suppose then
 $U=\langle E_1, E_3,E_6,E_{10}\rangle, V = \langle E_0,E_5,E_8 \rangle
 $. Then $U*V=\langle E_1,E_3,E_6,E_8, E_9,E_{10},E_{11},E_{14},E_{15},
 E_{18}\rangle$ and let
$C^\perp = U*V$. Then $C$ is a code with distance $\geq 2t+1=7$. Actually
$C^\perp$ is an $(19,10,9)$ code and $C$ is an $(19,9,11)$ code. So in
fact the code $C$ is a $5$-error correcting code but we just have a
$3$-error correcting pair.  
\item Example 2.

Now let $A$ be as in Example 1. 
 Here we produce a $5$-error correcting pair by choosing 
randomly $U=\langle E_1,E_3,E_6,E_{10},E_{18}\rangle$ and then choosing $V$ to
be $4$ of these say $V=\langle E_1,E_3,E_6,E_{10}\rangle$. Then let
      $C^\perp= U*V$.
Now $U*V$ has $13$ elements and so $C^\perp$ is a $(19,13,6)$ code and
$C$ is a $(19,6,14)$ code. Thus $C$ is a $6$ error correcting code and
we have a $5$ error correcting pair for it.
\item Consider the Cauchy (which is Hilbert) 
matrix $A=\begin{pmatrix} 1& 1/2& 1/3 & 1/4 &
				     \ldots \\ 1/2 & 1/3 & 1/4 & 1/5 &
	   \ldots
\\ 1/3&1/4&1/5& 1/6 & \ldots \\ \vdots & \vdots & \vdots & \vdots &
	   \vdots \\ \end{pmatrix}$

Denote the rows of $A$ by $E_1, E_2, \ldots$. Suppose we want a $t$
      error correcting code. Let $U=\langle E_1,E_2,E_3\rangle, V
      =\langle E_1,E_2\rangle$. Then $U*V= \langle E_1*E_1,
      E_2*E_2,E_3*E_1, E_2*E_2,E_2*E_3\rangle$ and let $C^\perp = U*V$. 

Then 

$C^\perp = \langle (1,1/4,1/9,1/16, \ldots),
      (1/2,1/6,1/12,1/20,\ldots), \\ 
      (1/3,1/8,1/15, 1/24,\ldots),(1/4,1/9,1/16,
      1/25,\ldots),(1/6,1/12,1/20,1/30, \ldots)\rangle$

is the required code with which to take the `samples'. Now $C$ is an
      $(n,n-5)$ code (provided the elements in $C^\perp$ are
      independent) and is $2$-error correcting with error locating pair
      $(U,V)$; now $C$ may be a $(n,n-5,6)$ code but is by the theory
       a $(n,n-5,\geq 5)$ code. If the elements of $C^\perp$ are not
      independent then $C$ is an $(n,n-4,5)$ code. 
\end{itemize}

\subsection{Method} 
Suppose now $A$ is a matrix such that any square submatrix has non-zero
determinant.%%  and such that the rows of $A$ have the multiplicative
%% property $E_i*E_j=\al E_{i+j}$.

Now choose at random any $t+1$ rows of $A$ to form $U$ and then any
$t$ rows of $A$ to form $V$. Then let $C^\perp$ be the space generated
by $\{u*v | u\in U, v\in V\}$. From this it is deduced that $d(C)\geq
2t+1$ and $C$ has $t$-error correcting pair $(U,V)$. Then proceed
as before in Section \ref{solve} to produce the decoding algorithm
with the $t$-error correcting pair with which to solve $Aw=y$ where $w$
has at most $t$ non-zero entries and $E_jw$ are known for $E_j\in C^\perp$. 

We don't need the multiplicative property $E_i*E_j= E_{i+j}$ on the
rows of $A$ although $U*V$ could be large; the largest rank that $U*V$
could have is $t(t+1)$ but can often be made of a smaller order.   
%The worst number it can be by this random method is $t(t+1)$.
However selections can be made so that the resulting code has
dimension of $O(t)$. % where $s$ is a small fixed number. 
This for example by choosing
the rows in $U$, $|U|=t+1$, and in $V$, $|V|=t$ 
to be in arithmetic sequence with the same difference will give $C^\perp
=U*V$ with $2t$ elements; variations of the differences will also give
$|U*V|= st$ for very small $s$ (compared to $t$). %in their selection.  

Being able to randomly choose the error-correcting pairs and thus the
measurements $C^\perp$ suggests that encryption methods may possibly 
be introduced into the system.

Thus: \begin{enumerate}
\item  In $Aw$ it is given that $w$ has at most $t$ non-zero
entries and that  the determinant of any square
submatrix of $A$ is non-zero.
\item Choose $t+1$ rows of $A$ to form $U$ and then $t$ rows of $A$ to
  form $V$.
\item Let $C^\perp = U*V$. Then $C$ has distance $\geq 2t+1$ and
  $(U,V)$ is a $t$-error correcting pair for $C$. 
\item The measurements/samples $E_jw$ are taken for $E_j$ in a
  generating set of $C^\perp$.
\item The value of $w$ is then determined by the decoding methods of Section
  \ref{algol}. Details are omitted.
\end{enumerate}
%Details are omitted. 

\section{Determinants of Submatrices}
The Vandermonde matrix $V=V(x_1,x_2,\ldots,x_n)$ is defined by
   
$V=V(x_1,x_2,\ldots,x_n) = \begin{pmatrix}
1&1&\ldots &1 \\ x_1& x_2& \ldots& x_n \\ \vdots & \vdots &
\vdots & \vdots \\ x_1^{n-1} & x_2^{n-1} & \ldots &
x_n^{n-1} \end{pmatrix}$ 

It is well-known that the determinant of $V$ is non-zero if and only
if the $x_i$ are 
distinct. Assume the $x_i$ are non-zero. %%  in all cases with which we
%% deal. 

\begin{proposition} Let $V=V(x_1,x_2,\ldots, x_n)$ be a Vandermonde
 matrix with  rows and columns numbered $\{0, 1, \ldots, n-1\}$. 
 Suppose rows $\{i_1,i_2,\ldots,i_s\}$ and columns
 $\{j_1,j_2,\ldots, j_s\}$ are chosen to form an $s\times s$ submatrix $S$
 of $V$ and that $\{i_1,i_2, \ldots, i_s\}$ are in arithmetic
 progression with arithmetic difference $k$. Then

$$|S|= x_{k_1}^{i_1}x_{k_2}^{i_1} \ldots
 x_{k_s}^{i_1}|V(x_{k_1}^k,x_{k_2}^k, \ldots, x_{k_s}^k)|$$
\end{proposition}
\begin{proof} Note that $i_{l+1}-i_l= k$ for $l=1,2,\ldots, s-1$, for
 $k$ the fixed arithmetic difference.

Now $ S=\begin{pmatrix}x_{k_1}^{i_1} & x_{k_2}^{i_1} & \ldots &
       x_{k_s}^{i_1} \\ x_{k_1}^{i_2} & x_{k_2}^{i_2}&\ldots
       & x_{k_s}^{i_2}
\\ \vdots & \vdots & \vdots & \vdots \\ x_{k_1}^{i_s} & x_{k_2}^{i_s}&
       \ldots & x_{k_s}^{i_s} \end{pmatrix} $
 and so $|S| = \left|\begin{array}{cccc}x_{k_1}^{i_1} & x_{k_2}^{i_1} & \ldots &
       x_{k_s}^{i_1} \\ x_{k_1}^{i_2} & x_{k_2}^{i_2}&\ldots
       & x_{k_s}^{i_2}
\\ \vdots & \vdots & \vdots & \vdots \\ x_{k_1}^{i_s} & x_{k_2}^{i_s}&
       \ldots & x_{k_s}^{i_s} \end{array}\right| $.

Hence by factoring out $x_{k_i}$ from column $i$ for $i=1,2,\ldots, s$
it follows that  

$|S| = x_{k_1}^{i_1}x_{k_2}^{i_1}\ldots x_{k_s}^{i_1} \left|
 \begin{array}{cccc}1&1&\ldots &1 \\ x_{k_1}^k & x_{k_2}^k & \ldots & x_{k_s}^k
\\ x_{k_1}^{2k} & x_{k_2}^{2k} & \ldots &x_{k_s}^{2k} \\ \vdots &
  \vdots & \vdots & \vdots \\ x_{k_1}^{(s-1)k} & x_{k_2}^{(s-1)k} &
  \ldots & x_{k_s}^{(s-1)k}\end{array}\right|
= x_{k_1}^{i_1}x_{k_2}^{i_2}\ldots x_{k_s}^{i_s}|V(x_{k_1}^k,x_{k_2}^k,
 \ldots, x_{k_s}^k)|$

\end{proof}

A similar result holds when the columns $\{j_1,j_2,\ldots,
j_s\}$ are in arithmetic progression. 
\begin{corollary} $|S|\neq 0$ if and only if $|V(x_{k_1}^k,x_{k_2}^k,
 \ldots, x_{k_s}^k)|\neq 0 $. 
\end{corollary}
\begin{corollary} $|S| \neq 0$ if and only if $x_{k_i}^k \neq x_{k_j}^k$
 for $i\neq j, 1\leq i,j \leq s$. This happens if and only if
 $(x_{k_i}{x_{k_j}^{-1}})^k \neq 1$ for $i\neq j, 1\leq i,j \leq s$. 
\end{corollary}
\begin{corollary}\label{unity}  $|S|\neq 0$ if and
 only if $(x_{k_i}x_{k_j}^{-1})$ is not a $k^{th}$ root of unity for
 $i\neq j, 1\leq i,j\leq s$.
\end{corollary}
\begin{corollary}\label{real} If  the entries $\{x_1, x_2, \ldots, x_n\}$
  are real then 
 matrix $|S|\neq 0$ if either (i) $k$ is odd or (ii) $k$
 is even and $x_i\neq -x_j$ for $i\neq j$. %%  Now (ii) is equivaalent to
 %% saying that $k$ is even and $|x_i| \neq |x_j|$.
 \end{corollary}
\begin{corollary}\label{pos} If the entries $\{x_1, x_2, \ldots, x_n\}$ 
are real and positive then $|S|\neq 0$. 
\end{corollary}

\begin{corollary}\label{fouri} When $x_i=\om^{i-1}$ for a primitive
 $n^{th}$ root of unity $\om$ (that is, when $V$ is the Fourier
  $n\times n$ matrix) and $\gcd(k,n)=1$ then $|S|\neq 0$. 
\end{corollary}
\begin{proof} If $(x_{k_1}/x_{k_j})^k= 1$ then
  $(\om^{k_i-1}\om^{1-k_j})^k=1$ and so $\om^{k(k_i-k_j)} = 1$. 
  As $\om$ is a primitive $n^{th}$ root of unity this implies 
    that $k(k_i-k_j) \equiv 0 \mod n$. As $\gcd(k,n)=1$ this implies
  $k_i-k_j \equiv 0 \mod n$ in which case $k_i=k_j$ as $1\leq k_i<n,
  1\leq k_j < n$. 
\end{proof}

\section{Vandermonde matrices}\label{vander}

%Vandermonde real matrix
%has the property that the determinant of any square submatrix is
%non-zero. 
Let 
$A=V(x_1,x_2,\ldots, x_{n})$ be a  Vandermonde with rows $\{E_0,E_1,\ldots, E_{n-1}\}$. Then
$E_i*E_j=E_{i+j}$. As in Section \ref{rows} let $ C^\perp =
\langle E_{j_1}, E_{j_2}, \ldots,
E_{j_u}\rangle$. By  Corollary \ref{unity} if $C^\perp $ has rows in
arithmetic sequence with arithmetic difference $k$ 
and the ratios $x_i/x_j$ for $i\neq j$ in $A$ are not $k^{th}$ roots
of unity then $C$ (the dual of $C^\perp$) is an $(n,n-2t,2t+1)$ code and is $t$-error
correcting with  $C^\perp$ as the check matrix. As shown in Section
\ref{solve},  $C$ has
an error correcting pair and  Algorithm \ref{algor1} in Section
\ref{algol} may be applied.   

%% If the matrix has real entries then automatically
%% $x_i/x_j$, for $i\neq j$, is not a $k^{th}$ root of unity.

Thus Vandermonde matrices for which $x_i/x_j$ are not roots of unity 
are obvious choices in which to take rows of the matrix which are
evenly spaced.  
Then Theorem \ref{correct3} is satisfied and   
 the decoding Algorithm \ref{algor} or \ref{algor1}
solves the underdetermined system $Aw=y$ with Vandermonde matrix, provided
the number of non-zero entries of $w$ is limited. 

Consider then a Vandermonde matrix

$V=V(\all_1,\all_2,\ldots,\all_n) = \begin{pmatrix}
1&1&\ldots &1 \\ \all_1& \all_2& \ldots& \all_n \\ \vdots & \vdots &
\vdots & \vdots \\ \all_1^{n-1} & \all_2^{n-1} & \ldots &
\all_n^{n-1} \end{pmatrix}$ 

We assume the $\all_i$ are distinct and non-zero.

 Define
$E_k$ to be  $(\all_1^k,\all_2^k, \ldots, \all_n^k)$ for any $k\in
 \Z$. The rows of $V$ are $\{E_0,E_1,\ldots, E_{n-1}\}$.  
\begin{lemma}\label{prod} $E_i*E_j = E_{i+j}$.
\end{lemma} 

Thus we obtain the following set-up. Let $A=V(\all_1,\all_2,\ldots,
\all_{n})$ and $Aw=y$. Measurements $E_jw$ (values of $y$) are taken or
known for $j \in M= \{j_1, j_2, \ldots, j_u\}\subset \{0,1,\ldots, n-1\}$  where $u\geq 2t$. 
The elements in $M$ are in arithmetic
progression with difference $k$ and $\all_i/\all_j$ is not a $k^{th}$
root of unity for $i\neq j$. % indeed it is sufficient that
                             % \all_i/\all_j not be $k^{th}$  root sof
                             % unity for $i,j \in J$ and $i\neq j$  

The following Algorithm \ref{algor10} finds $w$; this is special case of
Algorithm \ref{algor1} but can be read here independently of this.

Define $F_i=E_{j_i}$ for $j_i \in J$ %% where we assume different
%% meanings for $E_k$ and $E_{k_l}$.
%and $F_0=E_{j_1-k}$. 
% We assume here that $j_1\geq k$; the case
% $j_i<K$ can be dealt with. % with indices taken$\mod n$.
Let $\al_i=<w,E_{j_i}> = E_{j_i}w$ for $j_i\in J$. Let $F_i=E_{j_i}$
for $j_i\in J$.
Thus $\al_i=<w,F_i>$.

\begin{Algorithm}\label{algor10} %[ref=(\alph*)]

\quad
\begin{enumerate}[label=(\roman*),ref=(\roman*)]

\item Find a non-zero  element $v\T$ of the kernel of 
$E=\begin{pmatrix} \al_1 &\al_2 & \ldots & \al_{t+1} \\ \al_2 & \al_3 &
  \ldots & \al_{t+2} \\ \vdots & \vdots & \vdots & \vdots \\ \al_t &
  \al_{t+1} & \ldots & \al_{2t}\end{pmatrix}$.\label{eqs64}
\item Let $a= (F_1,F_2,\ldots, F_{t+1})v\T$. 
\item Suppose $v\T = (v_1,v_2,  \ldots, v_{t+1})\T$. 
The $i^{th}$ component of
  $a$ is $(v_1\all_i^{j_1} + v_2\all_i^{j_1+k} + \ldots + v_{t+1}
\all_i^{j_1+tk})$; we are interested in when this  is $0$. 

The $i^{th}$ component of $a$ is $0$ if and only if 
$v_1+v_2\all_i^k+v_3\all_i^{2k}+ \ldots + v_{t+1}\all_i^{tk} = 0.$
\label{eqs24}
% (Any non-zero multiple of $a$
%       will suffice as we are only interested  in the zero entries of
%       $a$. Note that $a$ is a $1\times n$ vector.)
\item Let $z(a)= \{j | 
a_j=0\}$  which is the set of locations of the zero coordinates of
$a$. 
Suppose $z(a)=\{i_1,i_2, \ldots, i_t\}$ and denote this set by $J$.\label{eqs30}
\item Solve $s_J(x) = s(w)$. This reduces to solving the following:
\begin{eqnarray}\label{est10} %[label=(\Alph*),ref=(\Alph*)]
\begin{pmatrix}\all_{i_1}^{j_1} & \all_{i_2}^{j_1} & \ldots
  &\all_{i_t}^{j_1} \\ \all_{i_1}^{j_2} &
 \all_{i_2}^{j_2} & \ldots &\all_{i_t}^{j_2} \\ \vdots & \vdots & \vdots &\vdots \\ 
\all_{i_1}^{j_{2t}} & \all_{i_2}^{j_{2t}} & \ldots
&\all_{i_t}^{j_{2t}} \end{pmatrix}\begin{pmatrix}x_1 \\ x_2 \\ \vdots \\ x_t
			\end{pmatrix}
= \begin{pmatrix} \al_1 \\ \al_2 \\ \vdots \\ \al_{2t} \end{pmatrix}   
\end{eqnarray}

Since the elements in $M$ have arithmetic difference $k$ so that $j_s=
i_1+(s-1)k$ for $1\leq s \leq 2t$, this equation (\ref{est10}) is equivalent to 
\begin{eqnarray}\label{est11} %[label=(\Alph*),ref=(\Alph*)]
\begin{pmatrix} 1 & 1 & \ldots
  & 1 \\ \all_{i_1}^{k} &
 \all_{i_2}^{k} & \ldots &\all_{i_t}^{k} \\ \vdots & \vdots & \vdots &\vdots \\ 
\all_{i_1}^{(2t-1)k} & \all_{i_2}^{(2t-1)k} & \ldots
&\all_{i_t}^{(2t-1)k} \end{pmatrix}\begin{pmatrix}\be_{i_1}^{j_1}x_1
  \\ \be_{i_2}^{j_1}x_2 \\ \vdots  \\ \be_{i_t}^{j_1}x_t
			\end{pmatrix}
= \begin{pmatrix} \al_1 \\ \al_2 \\ \vdots \\ \al_{2t} \end{pmatrix}   
\end{eqnarray}
\item 
Then $x=(x_1,x_2,\ldots,x_t)$ is obtained from these equations
(\ref{est11}) (or from (\ref{est10})) and $w$ has   
 entries $x_i$ in positions  as determined by $J$ and zeros elsewhere.
\end{enumerate}
 
\end{Algorithm}

The matrix in (\ref{est11}) is a Vandermonde matrix. It
is sufficient to solve the first $t$ equations and the inverse of such a
$t\times t$ Vandermonde type matrix may be obtained in $O(t^2)$
operations.  In connection with item \ref{eqs64}, finding a non-zero
element of the kernel of a Hankel $t\times (t+1)$ matrix can be done in $O(t^2)$
operations.  

In connection with item \ref{eqs24}, consider
$f(x)=v_1+v_2x+v_3x^2+ \ldots + v_{t+1}x^{t}$. It is required to find
those $\all_i$ for which $f(\all_i^k)=0$. By Horner's method
$f(\all_i^k)$ may be determined in $O(t)$ operations and thus finding
all $i$ for which $f(\all_i^k)=0$ can be done in $O(nt)$
operations. Choose $j\in J$ for item \ref{eqs30} if
$f(\all_j^k)=0$. Finding the zeros of $f(x)$ takes the maximum of
$O(nt)$ operations and all other operations take a maximum of $O(t^2)$
operations. 
 
Calculations with Vandermonde type matrices obtained from the Fourier
matrix are known to be stable. 
\subsubsection{Which are best?}
A question then is which Vandermonde matrices are best for working with 
Algorithm \ref{algor10}. The Fourier matrix cases, which have entries in $\cc$, 
are dealt with in Section \ref{Fourier}. 

Which Vandermonde real matrices are best?

Possibilities  for investigation include 
\\ $V=V(\frac{1}{2},\frac{1}{3},\ldots,\frac{1}{n}) = \begin{pmatrix}
1&1&\ldots &1 \\ \frac{1}{2}& \frac{1}{3}& \ldots& \frac{1}{n} \\ \vdots & \vdots &
\vdots & \vdots \\ (\frac{1}{2})^{n-1} & (\frac{1}{3})^{n-1} & \ldots & 
(\frac{1}{n})^{n-1} \end{pmatrix}$ and 
\\ $V=V(\al,\al^2,\ldots,\al^n) =\begin{pmatrix} 1&1&\ldots& 1 \\ \al
&\al^2 & \ldots & \al^n \\ \vdots &\vdots & \vdots &\vdots
\\ \al^{n-1} & \al^{2(n-1)} &\ldots & \al^{n(n-1)} \end{pmatrix}$  
where $\al^i\al^{-j}$ is not a $k^{th}$ root of unity for $i\neq j, 1\leq i,j \leq n$. 
%\subsection{Fourier}
\section{Fourier matrix}\label{Fourier}  
%Consider when $A$ is a Fourier $n\times n$ matrix. 

 Suppose now that $A$ is the  Fourier $n\times n$ matrix with rows $\{E_0,E_1,\ldots,E_{n-1}\}$.  Measurements are taken of $Aw$, that is
certain $E_jw$ are taken or known for $j \in J= \{j_1, j_2, \ldots, j_u\}$   
and it is given that $u\geq 2t$ where $t$ is the maximum number of
non-zero entries of $w$. 

\begin{theorem}\label{arith} Suppose the $E_j$ in $C^\perp = \langle
  E_{j_1}, E_{j_2}, \ldots, E_{j_u}\rangle $ are evenly spaced with
  arithmetic difference $k$ satisfying $\gcd(n,k)=1$. Then  any $u\times
  u$ square submatrix of $\hat{C}$ has non-zero determinant.
\end{theorem}
\begin{proof}

This follows directly from Corollary \ref{fouri}.
\end{proof} 

%\end{theorem}
Proposition 7 of \cite{new} may also be  used to prove Theorem
\ref{arith} above. 
%% to get the mds
%% property required when the measurements taken are evenly spaced with
%% arithmetic difference $k<n$ satisfying $\gcd(n,k)=1$. This will include
%% all cases when $n$ is prime.  
%% \begin{proposition}\label{junki}(Proposition 7 of \cite{new}).
%% Let $\{\om_1, \om_2, \ldots, \om_t\}$ be distinct $n^{th}$ roots of unity
%% and let $\{k_1,k_2, \ldots, k_t\}$ be an arithmetic progression of
%% integers with common difference $k$ satisfying $\gcd(n,k)=1$. Then
%% $\det{\om_i^{k_i}}\neq 0$.
%% \end{proposition}
This Proposition 7 of \cite{new} is analogous to Chebotar\"ev's theorem.%  (see for
% example \cite{tao} but it is proved  in a number of papers).  

%% \begin{proof} This follows directly from Proposition \ref{junki}.
%%  The $u\times u$
%%   subdeterminants of $\hat{C}$ are transposes of those of this
%%  Proposition \ref{junki} and are thus non-zero. 
%% \end{proof}
\begin{corollary}\label{coro} Let $\C$ be the code with check matrix from 
$C^\perp = \langle
  E_{j_1}, E_{j_2}, \ldots, E_{j_{u}}\rangle $ where the $E_{i_j}$ are evenly spaced 
with arithmetic
 difference $k$ satisfying $\gcd(n,k)=1$. Then $\C$ is an mds
 $(n,n-u,u+1)$ code.
\end{corollary}

Consider cases where $u> 2t$. Here $\hat{C}$ is a $(n,u)$ matrix and $C$
is a $(n,n-u)$ matrix. It is required that $\mathcal{C}$ be a $t$-error
correcting code and thus it is required that $C$ be a $(n,n-u, \geq(2t+1))$
code. For this to happen it is required that any $2t$ columns of $\hat{C}$ be
linearly independent. 

Let $A$ be Fourier $n\times n$ matrix with rows $\{E_0,E_1,\ldots,
E_{n-1}\}$. When $n$ is prime the code generated by $\langle E_{j_1},
E_{j_2}, \ldots, 
E_{j_u}\rangle$ is an $(n,u,n-u+1)$ code; see \cite{hur1}. 
In this case then $C^\perp = \langle
E_{j_1}, E_{j_2}, \ldots, 
E_{j_u}\rangle$ with $u=2t$ generates  an $(n,2t,n-2t+1)$ code and
$\mathcal{C}$, its dual, is an
$(n,n-2t,2t+1)$ code. Thus $\mathcal{C}$ is a $t$-error correcting code. Now it is
required to find a decoding algorithm for $w$ as  an error word of this
code.

Assume that the $E_{j_k}$ are
evenly distributed, that is, 
$C^\perp = \langle E_i, E_{i+j}, E_{i+2j},
 \ldots E_{i+(2t-1)j}\rangle$ where suffices are taken $\mod n$.  %% We
 %% show that in this case it is easy to get a $t$-error correcting pair,
 %% definition \ref{error}, when $\gcd(n,j)=1$. 
 %% In these cases the vector $w$ which has at most $t$ non-zero entries may be
 %% obtained by the method of Algorithm \ref{pell} above due to Pellikaan
 %% \cite{pell1}. It will be shown that the solution may be obtained in $O(t^2)$
 %% operations. 

%%  Take the case  $C^\per
%% = \langle E_1,E_2,\ldots, E_{2t}\rangle$, 
%% that is,  the $E_i$ are consecutive starting at $E_1$; the more
%% general case will be dealt with similarly.%  but the
% process works in general. % situation. %% 

% Theorems \ref{correct2} and \ref{correct3} may now be applied to
% the Fourier matrix in this situation. This then gives the algorithms
% for this case as follows: 

\subsection{Algorithm for Fourier}\label{algol1}
This Algorithm is a special case of previous algorithms but can be
read here independently. 

Suppose $y=Ax$ where 
 $A$ is an 
$n\ti n$  Fourier matrix, $w$ an $n\ti 1$ unknown vector and 
where $u$ entries of $y$ are known. It is given that %It must be assumed 
$w$ has at most $t$ non-zero entries and that  $u\geq 2t$. 
Denote the rows of $A$ by $\{E_0,E_1,\ldots, E_{n-1}\}$. % so that 

%% In summary then: $w$ is size $n$ vector with at most $t$ non-zero
%% entries.
Measurements $E_jw$ (values of $y$) are taken or known for $j \in J= \{j_1, j_2,
\ldots, j_u\}$  where $u\geq 2t$. 
%% Suppose the measurements $E_1w, E_2w, \ldots, E_{2t}w$
%% are taken.
We give  an Algorithm to calculate the value of $w$ when the
measurements are in an arithmetic progression (evenly distributed)
with difference $k$ satisfying $\gcd(n,k)=1$.

\subsubsection{Case $k=1$}
We first for clarity give the algorithm when $K=\{1,2,\ldots, 2t\}$. This
is easier to explain and avoids the complicated notation necessary for
the general case given below. The results and algorithm obtained in
this case, where the measurements are taken consecutively, 
are similar to those  in \cite{FRI}.  

The set-up then is that $A$ is the Fourier $n\times n$ matrix with rows
$\{E_0,E_1,\ldots, E_{n-1}\}$ and that
measurements $E_iw$ are taken for $i=1,2,\ldots, 2t$. It is assumed
that $w$ has at most $t$ non-zero entries. 
Then $w$ is  determined as follows:
Let $\al_i=<w,E_i> = E_iw$ for $i\in J=\{1,2,\ldots,2t\}$. 
% Let $\be_j = <w,E_j> = E_jw$ for $j=\in J$  and define
% $\al_i=\be_{j_i}$. 

\begin{Algorithm}\label{algor2}

\quad
\begin{enumerate}

\item Find a non-zero  element $x\T$ of the kernel of 
$E=\begin{pmatrix} \al_1 &\al_2 & \ldots & \al_{t+1} \\ \al_2 & \al_3 &
  \ldots & \al_{t+2} \\ \vdots & \vdots & \vdots & \vdots \\ \al_t &
  \al_{t+1} & \ldots & \al_{2t}\end{pmatrix}$.
\item Let $a= (E_1,E_2,\ldots, E_{t+1})x\T$. (Any non-zero multiple of $a$
      will suffice as we are only interested  in the zero entries of
      $a$. Note that $a$ is a $1\times n$ vector.)

\item Let $z(a)= \{j | 
a_j=0\}$  which is the set of locations of the zero coordinates of
$a$. Suppose $z(a)=\{i_1,i_2, \ldots, i_t\}$ and denote this set by $J$.
\item Solve $s_J(x) = s(w)$. This reduces to solving the following:

\begin{eqnarray}\label{four1}\begin{pmatrix}\om^{i_1-1} & \om^{i_2-1} & \ldots &\om^{i_t-1}\\
      \om^{2(i_1-1)}  &
 \om^{2(i_2-1)} & \ldots &\om^{2(i_t-1)} \\ \vdots & \vdots & \vdots \\ 
\om^{2t(i_1-1)} & \om^{2t(i_2-1)} & \ldots
&\om^{2t(i_t-1)}\end{pmatrix}\begin{pmatrix}x_1 \\ x_2 \\ \vdots \\ x_t
			\end{pmatrix}
= \begin{pmatrix} \al_1 \\ \al_2 \\ \vdots \\ \al_{2t} \end{pmatrix}   
\end{eqnarray}

This is equivalent to solving the following:
\begin{eqnarray}\label{four2} \begin{pmatrix}1 & 1 & \ldots & 1\\
      \om^{2(i_1-1)}  &
 \om^{2(i_2-1)} & \ldots &\om^{2(i_t-1)} \\ \vdots & \vdots & \vdots \\ 
\om^{2t(i_1-1)} & \om^{2t(i_2-1)} & \ldots
&\om^{2t(i_t-1)}\end{pmatrix}\begin{pmatrix}\om^{i_1-1}x_1
  \\  \om^{i_2-1}x_2 \\ \vdots \\ \om^{i_t-1}x_t 
			\end{pmatrix}
= \begin{pmatrix} \al_1 \\ \al_2 \\ \vdots \\ \al_{2t} \end{pmatrix}
\end{eqnarray}   

\item 
The value of $w$ is then obtained from  the solution $x=(x_1,
x_2,\ldots,x_t)$ of equations (\ref{four2}) (or equations
(\ref{four1})) with entries $x_i$ in positions as determined by $J$ and
zero elsewhere.
\end{enumerate}
 
\end{Algorithm}
The complexity of the operations is discussed in Section \ref{complexity}.

\subsubsection{General Fourier case}\label{general1} 
%\section{Algorithm}
Suppose 
 $A$ is an 
$n\ti n$  Fourier matrix. %  $w$ an $n\ti 1$ unknown vector and 
% where $u$ entries of $y$ are known. It is given that %It must be assumed 
% $w$ has at most $t$ non-zero entries. % and that  $u\geq 2t$.
 Denote the rows of $A$ by $\{E_0,E_1,\ldots, E_{n-1}\}$. % so that 

%% In summary then: $w$ is size $n$ vector with at most $t$ non-zero
%% entries.
Measurements $E_jw$ are taken or known for $j \in M= \{j_1, j_2,
\ldots, j_{2t}\} \subset \{0,1,\ldots, n-1\}$. The elements in $M$ are in arithmetic
progression with difference $k$ satisfying $\gcd(n,k)=1$. Thus
$j_s=j_1+(s-1)k$ for $s=1,2,\ldots, 2t$. 
%% Suppose the measurements $E_1w, E_2w, \ldots, E_{2t}w$
%% are taken.
Then $w$ is calculated by the following algorithm.
% We give  an Algorithm to calculate the value of $w$ when the
% measurements are in an arithmetic progression (evenly distributed)
% with difference $k$ satisfying $\gcd(n,k)=1$.

% \subsection{Case $k=1$}
% We first for clarity give the algorithm when $J=\{1,2,\ldots, u\}$. This
% is easier to explain and avoids the complicated notation necessary for
% the general case given below.  
% The set-up then is that $A$ is the Fourier $n\times n$ matrix with rows
% $\{E_0,E_1,\ldots, E_{n-1}\}$ and that
% measurements $E_iw$ are taken for $i=1,2,\ldots, 2t$. It is assumed
% that $w$ has at most $t$ non-zero entries. 
% Then $w$ is  determined as follows:
%Let $\al_i=<w,E_i> = E_iw$ for $i\in J=\{1,2,\ldots,2t\}$. 
Let $\al_k = <w,E_{j_k}> = E_{j_k}w$ for $j_k\in J$.   % and define
% $\al_i=\be_{i}$. 
Define $F_i=E_{j_i}$ for $j_i \in J$. Thus $\al_k=<w,F_k> = F_kw$. %% where we assume different
%% meanings for $E_k$ and $E_{k_l}$.
%and $F_0=E_{j_1-k}$ with indices taken$\mod n$.  

\begin{Algorithm}\label{algor3}

\quad
\begin{enumerate}[label=(\roman*),ref=(\roman*)]

\item Find a non-zero  element $v\T$ of the kernel of 
$E=\begin{pmatrix} \al_1 &\al_2 & \ldots & \al_{t+1} \\ \al_2 & \al_3 &
  \ldots & \al_{t+2} \\ \vdots & \vdots & \vdots & \vdots \\ \al_t &
  \al_{t+1} & \ldots & \al_{2t}\end{pmatrix}$.
\item Let $a= (F_1,F_2,\ldots, F_{t+1})v\T$. % (Any non-zero multiple of $a$
%       will suffice as we are only interested  in the zero entries of
%       $a$. Note that $a$ is a $1\times n$ vector.)
\item Suppose $v\T = (v_1,v_2,  \ldots, v_{t+1})\T$. 
The $i^{th}$ component of
  $a$ is $v_1\om^{(j_1)(i-1)} + v_2\om^{(j_1+k)(i-1)} + \ldots + v_{t+1}
\om^{(j_1+tk)(i-1)}$; we are interested in when this  is $0$. 

The $i^{th}$ component of $a$ is $0$ if and only if 
$v_1+v_2\om^{(i-1)k}+v_3\om^{(i-1)2k}+ \ldots + v_{t+1}\om^{(i-1)tk} = 0.$
\label{eqs34}

\item Let $z(a)= \{j | 
a_j=0\}$  which is the set of locations of the zero coordinates of
$a$. Suppose $z(a)=\{i_1,i_2, \ldots, i_t\}$ and denote this set by $J$.
\item Solve $s_J(x) = s(w)$. This reduces to solving the following:

\begin{eqnarray}\label{matrix}\begin{pmatrix}\om^{j_1(i_1-1)} &
			       \om^{j_1(i_2-1)} & 
\ldots &\om^{j_1(i_t-1)}\\ \om^{j_2(i_1-1)} &
 \om^{j_2(i_2-1)} & \ldots &\om^{j_2(i_t-1)} \\ \vdots & \vdots & \vdots \\ 
\om^{j_{2t}(i_1-1)} & \om^{j_{2t}(i_2-1)} & \ldots
&\om^{j_{2t}(i_t-1)}\end{pmatrix}\begin{pmatrix}x_1 \\ x_2 \\ \vdots \\ x_t
			\end{pmatrix}
= \begin{pmatrix} \al_1 \\ \al_2 \\ \vdots \\ \al_{2t}
  \end{pmatrix}\end{eqnarray}

Since $j_s=j_1+k(s-1)$ this reduces to solving the following system of
equations: 
\begin{eqnarray}\label{matrix1}\begin{pmatrix}1 &
			       1 & 
\ldots & 1\\ \om^{k(i_1-1)} &
 \om^{k(i_2-1)} & \ldots &\om^{k(i_t-1)} \\ \om^{2k(i_1-1)} &
				\om^{2k(i_2-1)} & \ldots &
				\om^{2k(i_t-1)}\\ \vdots & \vdots & \vdots \\ 
\om^{(2t-1)k(i_1-1)} & \om^{(2t-`)k(i_2-1)} & \ldots
&\om^{(2t-1)k(i_t-1)}\end{pmatrix}\begin{pmatrix}\om^{j_1(i_1-1)}x_1
    \\ \om^{j_1(i_2-1)}x_2 \\ \vdots \\ \om^{j_1(i_t-1)}x_t
			\end{pmatrix}
= \begin{pmatrix} \al_1 \\ \al_2 \\ \vdots \\ \al_{2t}
  \end{pmatrix}\end{eqnarray}

\item 
Then $x=(x_1,x_2,\ldots, x_t)$ is obtained from these equations
(\ref{matrix1}) from which $w$ is derived with entries $x_i$ in
positions as determined by $J$.
\end{enumerate}
 
\end{Algorithm}

\subsubsection{Complexity}\label{complexity} Finding the kernel of $E$ is of $O(t^2)$ as
it involves finding the kernel of a $t\ti (t+1)$ Hankel matrix, which
is this case has dimension $1$ in order to satisfy the given conditions. 
Superfast  algorithms of
$O(t\log^2t)$ with which to find the kernel of  a Hankel matrix have been
proposed.  %% What is required of course is just
%% any element of the kernel and then find its zeros.
%The kernel has dimension $1$ in order to satisfy the given conditions. 
 Item \ref{eqs34}, as already pointed out in Section \ref{vander},
  can be done in $O(tn)$
 operations; however by considering  
 a Fourier Transform of
 $(v_1,v_2,\ldots, v_t, 0, \ldots, 0)$ it  can be performed in
 $O(n\log n)$ operations by a Fast Fourier Transform.

%Solving (\ref{matrix}) directly is of $O(t^3)$. 
The matrix (\ref{matrix}) is a special Vandermonde 
type involving roots of unity only. There is a
formula for the inverse of  any Vandermonde matrix,  the Bj\"ork, Pereyra method 
 \cite{vander}, which involves  $O(t^2)$ operations. %% and since the entries of the matrix in (\ref{matrix}) are roots
%% of unity the formula takes  which are stable
%% to implement.
Finding the inverse of
a Vandermonde matrix with roots of unity is known to be particularly stable.
 The method of  Bj\"ork, Pereyra \cite{vander} 
involves divisions by $(\al_i-\al_j)$ where $\al_i\neq \al_j$ and in
these cases the $\al_k$ are roots of unity. %%  and indeed divisions may be
%% kept until the end. 
% Moreover the algorithm due to  may be
% efficiently used as the entries of $B$ are $n^{th}$ roots of unity. %type and can be dealt with as in \cite{blahut}. 
% The matrix is a Vandermonde matrix involving roots of unity.
  The system (\ref{matrix}) or (\ref{matrix1}) could also be solved using 
%This formula does reduce to
the Forney Algorithm/formula, see   \cite{blahut} Chapter 7.

\end{document}